\documentclass[runningheads,a4paper]{article}

\usepackage{amssymb,amsmath,amsthm}
\usepackage{color}
\usepackage[mathcal]{euscript}
\usepackage{microtype}
\usepackage{array}
\usepackage{graphicx}
\usepackage[margin=1in]{geometry}
\usepackage{hyperref}
\usepackage[hang,flushmargin,bottom]{footmisc}
\usepackage[numbers,sort&compress]{natbib}
\usepackage[arxiv]{optional}

\graphicspath{{./fig/}}

\usepackage{amssymb,amsmath,amsthm}
\usepackage[mathcal]{euscript}
\usepackage{microtype}
\usepackage{array}
\usepackage{mathtools}

\def\eps{\varepsilon}

\newcommand{\norm}[1]{\left\lVert#1\right\rVert}
\def\seq#1{\langle #1 \rangle}
\def\abs#1{\left| #1 \right|}
\def\Paren#1{\left( #1 \right)}	
\def\Set#1{\left\{ #1 \right\}}

\def\ceil#1{\left\lceil #1 \right\rceil}

\newcommand{\C}	{\chi}
\def\unhappy#1{\text{uhp}({#1})}
\def\EMPH#1{\textbf{\boldmath #1}}
\usepackage{enumitem}
\setlist[description]{font=\itshape}

\def\notekm#1{\noindent \textcolor{teal}{Note: #1}}

\usepackage[dvipsnames,usenames]{xcolor}
\usepackage{pgfplots}
\pgfplotsset{compat=1.12}
\usepgfplotslibrary{fillbetween}
\usepackage{lmodern}
\usepackage{slantsc}
\usetikzlibrary{matrix,backgrounds,patterns}

\def\eps{\varepsilon}

\title{Locally Fair Partitioning}
\author {
    Pankaj K.~Agarwal \and 
    Shao-Heng Ko \and 
    Kamesh Munagala \and 
    Erin Taylor \and 
    \small{Department of Computer Science, Duke University}\\
    \small{\{pankaj, sk684, kamesh, ect15\}@cs.duke.edu}
}

\makeatletter
\renewcommand\subparagraph{%
 \@startsection {subparagraph}{5}{\z@ }{3.25ex \@plus 1ex
 \@minus .2ex}{-1em}{\normalfont \normalsize \bfseries }}%
\makeatother

\newtheorem{theorem}{Theorem}
\newtheorem{corollary}[theorem]{Corollary}
\newtheorem{lemma}[theorem]{Lemma}
\theoremstyle{definition}
\newtheorem{definition}[theorem]{Definition}

\def\short#1{}

\begin{document}

\maketitle
\begin{abstract}
We model the societal task of redistricting political districts as a partitioning problem: Given a set of $n$ points in the plane, each belonging to one of two parties, and a parameter $k$, our goal is to compute a partition $\Pi$ of the plane into regions so that each region contains roughly $\sigma = n/k$ points. $\Pi$ should satisfy a notion of ``local'' fairness, which is related to the notion of core, a well-studied concept in cooperative game theory. A region is associated with the majority party in that region, and a point is {\it unhappy} in $\Pi$ if it belongs to the minority party. A group $D$ of roughly $\sigma$ contiguous points is called a {\it deviating} group with respect to $\Pi$ if majority of points in $D$ are unhappy in $\Pi$. The partition $\Pi$ is {\it locally fair} if there is no deviating group with respect to $\Pi$.

This paper focuses on a restricted case when points lie in $1$D. The problem is non-trivial even in this case. We consider both adversarial and ``beyond worst-case" settings for this problem. For the former, we characterize the input parameters for which a locally fair partition always exists; we also show that a locally fair partition may not exist for certain parameters. We then consider input models where there are ``runs'' of red and blue points. For such clustered inputs, we show that a locally fair partition may not exist for certain values of $\sigma$, but an approximate locally fair partition exists if we allow some regions to have smaller sizes. We finally present a polynomial-time algorithm for computing a locally fair partition if one exists.
\end{abstract}

\section{Introduction}
\label{s:intro}
Redistricting  is a common societal decision making problem. In its basic form, there are two parties, say red and blue, and a parliament with some $k$ representatives. Each individual (or voter) in the geographic region is aligned with one of the two parties. The goal is to divide the region into $k$ parts -- called districts -- so that each part elects one representative to the parliament. It is typically assumed that each district does majority voting, so that if a district has more red voters than blue voters, then the chosen representative will be red. 

The societal question then is {\it how should these districts be drawn}? One natural constraint is that the district is a connected region and, more preferably, has a compact shape. Another consideration is that each district is population {\em balanced}, i.e., has roughly the same number of individuals.\footnote{US courts have ruled that districts be population balanced, compact, and contiguous; see, e.g., \url{https://en.wikisource.org/wiki/Reynolds_v._Sims}.} A final consideration, and one that will be the focus of this paper, is {\em fairness}. If society has a large fraction of blue voters, the districts should not be drawn so that most representatives end up being red. 

In this paper, we consider a {\em local} and strong notion of fairness. Let us say that a voter is unhappy if she is in a majority blue (resp.~red) district, but her party is red (resp.~blue). 
We say that a given set of districts is \textit{locally fair} if no subset of unhappy voters of the same party can deviate and form a feasible district (nicely shaped and balanced) so that they are the majority in that district. In other words, these deviating voters have a \textit{justified complaint} -- there was a different hypothetical district where they could have been happy. This notion is akin to that of the {\it core} from cooperative game theory~\cite{Scarf}. As such, if a partition is locally fair, then it is \textit{as fair as possible} to the relevant parties -- there are no groups could potentially form a region and do better. 

There are examples in which a group of voters have argued they have a justified complaint regarding the redistricting.  In the 2012 election in North Carolina, $13$ House seats were allocated, $4$ to Democrats and $9$ to Republicans.  In contrast, the percentage of voters who voted for a Democrat candidate was $50.60\%$.  The U.S. Court of Appeals ruled that two of the districts' boundaries in this map were unconstitutional due to gerrymandering and required new maps to be drawn.  Considering this map, it is clear there exist compact potential districts which could be considered a deviating group with respect to the districting. This case (Cooper v. Harris (2017)) is just one in a long line of judgements on the \textit{fairness} of districting plans in the U.S.\footnote{See also, Benisek v. Lamone (2018), Gill v. Whitford (2018), Rucho v. Common Cause (2019), etc.} The exhibition of deviating groups may help a political group or group of voters justify their complaint that a redistricting is unfair, and it may be effectively used in auditing proposed plans. 

In contrast, some input instances may exhibit ``natural gerrymandering'', when the distribution of the population prevents redistricting plans from being representative to all groups~\cite{borodin2018big}. 
For example, if the minority party had $40\%$ of the vote in total but the voters are uniformly distributed, it is unlikely that any deviating groups with a justified complaint would exist. 
In this case, one could argue no reasonable redistricting could ensure the minority group elects its fair share of the representatives. 
Thus, the notion of local fairness introduced in the paper allows us to distinguish between natural and artificial gerrymandering. In contrast, when a redistricting plan is not {\it globally} fair (proportionally representative), it is not clear whether any group has a justified complaint regarding the redistricting, or if it is an unavoidable consequence of the geometry of the map. 

\subsection{Our Results} 
In this paper, we study the existence and computation of a locally fair partition in the {\em one-dimensional case}, where we assume the $n$ voters lie on a line or a circle. A \emph{feasible} district or region is now an interval containing $\sigma = n/k$ voters. 
The ``niceness'' aspect is captured by the region being an interval, and the ``balance'' aspect is captured by the number of points in each interval being $n/k$. 
Even in this setting, we show that locally fair partitioning is surprisingly non-trivial, and leads to a rich space of algorithmic questions. 

\paragraph{Relaxed local fairness notions.} In the 1D case, regulating each interval to be containing exactly $\sigma$ voters is extremely restrictive, and it is relatively easy to show that even a balanced (not necessarily fair) partition need not exist. We will therefore allow ourselves to relax the interval size. 
We parameterize this by $\eps$, so that the number of voters in any allowable interval lies in  $\left[(1-\eps) \sigma, (1+\eps) \sigma\right]$. This also relaxes the number of intervals to be some number in $[\frac{k}{1+\eps}, \frac{k}{1-\eps}]$.\footnote{Many of our results extend to the setting in which the number of intervals must be exactly $k$.} Further, we also relax the notion of \emph{deviation}, so that if a subset of voters deviate, they need to become ``really happy'' -- they need to be a strict majority in the interval to which they deviate. We call this parameter $\beta \in [1/2,1]$, so that unhappy points only deviate to a new allowable interval if their population size is at least $\beta \sigma$. 
If a fair partition exists under such relaxations, we term it \emph{$(\eps,\beta)$-fair}.

Under these relaxations, our first set of results in Section~\ref{s:adversarial} characterizes the $(\eps, \beta)$ values for which a fair partition exists, and those where it may not. For $\eps \le 1/5$, we show a {\em sharp threshold} at $\beta = 1 - \eps$: When $\beta < 1-\eps$, for large enough values of $n$, there is an instance with no $(\eps,\beta)$-fair partition, while when $\beta \ge 1-\eps$, the simple strategy of creating uniform intervals 
is $(\eps, \beta)$-fair. If we restrict points to deviate only when the interval they create has exactly $\sigma$ points, this sharp threshold holds for all $\eps \le 1/3$. 
To interpret this result, when $\eps = 1/3$, this means there is a fair partition where all intervals have size in the range $\left[\frac{2}{3} \sigma, \frac{4}{3} \sigma\right]$, and no subset of unhappy points can create an interval with $\sigma$ points, where they form $2/3$-majority. Furthermore, there is an instance where the bound of $2/3$ on the majority cannot be reduced any further. 

\paragraph{Beyond worst-case.} The negative results above are adversarial: they need careful constructions of sequences of runs of red and blue points with precise lengths, so that any partitioning scheme that needs intervals of certain size to eventually straddle both red and blue points in a way that allows a deviating interval to take shape.  However, this is an artifact of the intervals needing almost precise balance, i.e., their lengths being approximately $\sigma$.  The next question we ask is: suppose we are allowed to place a small fraction $\alpha$ of points in intervals whose sizes can be smaller than $(1-\eps) \sigma$. In particular, we could construct intervals for these points so that they are all happy, preventing them from deviating; or we could think of it as eliminating these points. Then {\it is it possible to circumvent these lower bounds?} 

In Section~\ref{s:periodic}, we show that the above is indeed the case when the input sequences are reasonably benign. By ``benign'', we mean that the input is clustered, i.e., composed of runs of red and blue points of arbitrary lengths, as long as these lengths are lower bounded by some value $\ell$. This models phenomena like {\it Schelling segregation}~\citep{schelling,zhang2011tipping,immorlica2017exponential}, where individuals have a slight preference for like-minded neighbors and relocate to meet this constraint, which leads to ``runs'' of like-minded individuals.

For such input sequences, we show that as long as all runs are of length at least $\ell = 2 \sigma$, once we allow a small fraction $\alpha = O(\frac{1}{k})$ of the points to be placed in unbalanced regions, there is a locally fair partition even for the strictest setting $(\eps, \beta) = (0, 1/2)$: the remaining points are placed in intervals of size exactly $\sigma$, and no deviating interval of size $\sigma$ has a simple majority of unhappy points. 


\subparagraph{Efficient partitioning.} In Section~\ref{s:fair-dp}, we finally study the algorithmic question: given parameters $(\eps, \beta)$, decide whether a given input of length $n$ admits a $(\eps, \beta)$-fair partition. Note that the results so far have been worst-case existential results, and it is possible that even when $\beta < 1 - \eps$, many inputs would have an $(\eps, \beta)$-fair partition. The challenge in designing an algorithm is that a deviating interval could involve points from more than one interval in the partition. We resolve this via a dynamic programming algorithm whose running time is polynomial in $n$ for any $\eps \in [0,1/2]$.

\short{
\smallskip
Due to length constraints of the paper, many proofs and discussion can be found in the full version~\cite{arxiv_version}.}

\subsection{Related Work}

\paragraph{Fairness notions.} Proportionality is a classic approach to achieving fairness in social choice. In a proportional solution, different demographic slices of voters feel they have been represented fairly. This general idea dates back more than a century~\cite{Droop}, and has recently received significant attention~\cite{CC,Monroe,Brams2007,Brill,Sanchez,PJR2018}. In fact, there are several elections, both at a group level and a national level, that attempt to find committees (or parliaments) that provide approximately proportional representation. 

The notion of core from cooperative game theory~\cite{Scarf} represents the ultimate form of proportionality: every demographic slice of voters feel that they have been fairly represented and do not have incentive to deviate and choose their own solution which gives all of them higher utility. In the typical setting where these demographic slices are not known upfront, the notion of core attempts to be fair to all subsets of voters. Though the core has been traditionally considered in the context of resource allocation problems~\cite{lindahl,galeS,foley,TTC}, one of our main contributions is to adapt this notion in a non-trivial way to the redistricting problem.

\paragraph{Redistricting vs. clustering.} 

The redistricting problem is closely related to the clustering problem. In a line of recent work, various models of fairness have been proposed for center-based clustering.  
One popular approach to fairness ensures that each cluster contains groups in (roughly) the same proportion in which they exist in the population~\cite{chierichetti2018fair,zafar2017fairness}. 
The redistricting problem we consider may take the opposite view -- we effectively want the regions or clusters to be as close to monochromatic as possible to minimize the number of unhappy points in each region.

Chen et al. studied a variant of fair clustering problem where any large enough group of points with respect to the number of clusters are entitled to their own cluster center, if it is closer in distance to all of them~\cite{chen2019proportionally}. This extends the notion of the core in a natural way to clustering. However, this work defines happiness of a point in terms of its distance, while in the redistricting problem, the happiness is in terms of the color of the majority within that region. The latter leads to fundamentally different algorithmic questions.

\paragraph{Redistricting Algorithms.} 
There has been extensive work on redistricting algorithms, going back to 1960s \cite{hess1965nonpartisan}, for constructing contiguous, compact, and balanced districts. Many different approaches, including integer programming \cite{Goderbauer14}, simulated annealing \cite{altman2010promise}, evolutionary algorithms \cite{LiuCW16}, Voronoi diagram based methods \cite{svec2007applying, cohen2018balanced}, MCMC methods \cite{bangia2017redistricting, deford2021recombination}, have been proposed; see \cite{becker2020redistricting} for a recent survey.
A line of work on redistricting algorithms focuses on combating manipulation such as gerrymandering: when district plans have been engineered to provide advantage to individual candidates or to parties~\cite{borodin2018big}. 
For example, Cohen-Addad et al. propose a districting strategy with desirable geometric properties such as each district being a convex polygon with a small number of sides on average~\cite{cohen2018balanced}. 
Using similar methods, Wheeler and Klein argue that the political advantage of urban or rural voters tends to be dramatically less than that afforded by district plans used in the real world~\cite{wheeler2020impact}. 
In fact, Chen et al. show that district plans can also have unintentional bias arising from differences in geographic distribution of two parties~\cite{chen2013unintentional}.

\paragraph{Auditing.}
Another line of work in redistricting focuses on developing statistical tools to detect gerrymandering given a districting plan~\cite{herschlag2020quantifying}.
In many redistricting algorithms, existing methods generate maps without explicitly incorporating notions of fairness, but instead focusing on compactness.
Popular methods generate an ensemble of plans and compare the number of representatives each party gets in the generated maps with the number received under the actual proposed maps~\cite{becker2020redistricting}. In practice, political groups use many justifications for whether a plan is \textit{fair}, and our paper offers a new formal model which may be used for auditing-- arguing that various plans satisfy properties of fairness~\cite{procaccia2021compact,deford2021recombination}. Our work in contrast takes a more algorithmic approach -- given a natural definition of what a fair redistricting should look like, we show existence and computational results.

\section{Preliminaries}
\label{s:prelim}

Let $X$ be a set of $n$ points in $\mathbb{R}^1$, each colored red or blue, and let $\sigma \in [n]$ be a parameter called ideal population size.\footnote{We assume points lie on a line for simplicity; our results extend to the case of a ring in a straightforward manner.} 
We wish to construct a locally fair partition of $X$ into intervals so that all intervals have roughly $\sigma$ points. Only ordering of points in $X$ really matters, so we describe the input as a (binary) sequence $X = x_1, \dots, x_n$, where $x_i \in \{\mathsf{R}, \mathsf{B}\}$ represents the color of the $i$-th point on the line. 
  

Define $R \coloneqq \{ i \in [n] \mid x_i = \mathsf{R} \}$ and $B \coloneqq \{ i \in [n] \mid x_i = \mathsf{B} \}$ to be the subset of all red points and blue points, respectively.
An \emph{interval} is a contiguous sequence, defined by a pair of integers $i, j \in [n]$ and denoted as either $[i,j]$ (where both points $i$ and $j$ are included) or $(i,j]$ (where only $j$ is included).
For an interval $I \subset [n]$, let $\abs{I}$ denote its size, i.e., $\abs{[i,j]} = j-i+1, \abs{(i,j]} = j-i$.

\paragraph{An alternative way to describe the input.} Sometimes it is useful to re-describe the input as a series of alternating \textit{maximal} monochromatic intervals. When appropriate, we denote the input as $X = R_1, B_1, R_2, \ldots, R_{\eta}, B_{\eta}$, where each $R_j \subseteq R$ (resp. $B_j \subseteq B)$ is a maximal sequence of red (resp. blue) points, $R_j \neq \varnothing$ for $j > 1$, and $B \neq \varnothing$ for $j < \eta$.
For each $R_j, B_j$, it suffices to specify its size. 

\paragraph{Balanced Partition.}
We are interested in partitioning $[n]$ into pairwise-disjoint intervals, i.e., computing a partition $\Pi = \seq{\pi_1, \dots, \pi_T}$, where $T = \abs{\Pi}$, $\pi_t = (i_{t-1}, i_t]$ for all $t \in [T]$, and $0 = i_0 < i_1 < \ldots < i_T = n$.

We parameterize the population deviation in an interval using an input parameter $\eps$:
For $\eps \in [0, 1/2]$, an interval $\pi_t \in \Pi$ is called $\eps$-\EMPH{allowable} (or allowable for brevity) if it satisfies $ (1-\eps) \sigma \leq \abs{\pi_t} \leq (1+\eps) \sigma.$
The partition $\Pi$ is \EMPH{balanced} if each of its interval is allowable.
Note that a balanced partition may not always exist (take $n = 100$, $\eps = .01$, and $\sigma = 40$).
In the remainder of the paper, we assume that $\sigma$ is chosen such that a balanced partition exists.

For $\eps = 0$, each interval has exactly $\sigma$ points; and for $\eps = 1/2$, each interval contains between $\frac{\sigma}{2}$ and $\frac{3\sigma}{2}$ points. 
In principle, we can choose $\eps$ to be any value in $[0, 1]$; but as the value of $\eps$ increases, the sizes of intervals in the partition become increasingly unbalanced. 
At an extreme, when $\eps = 1$, every point could form its allowable interval. 
Thus, we only consider the setting of $\eps \in [0, 1/2]$, though most of our results extend to settings of larger $\eps$. 

For an interval $I \subseteq [n]$, it is sometimes more convenient to work with the \textit{ratio} $\abs{I}/\sigma$, which is $1$ for an interval of $\sigma$ points. 
We define the {\it measure} of $I$, denoted by $\norm{I}$, as $\norm{I} = \abs{I}/\sigma.$

\paragraph{Locally fair partition.} 
Next, we turn our focus to defining the notion of local fairness.
For an interval $I$, let $\C(I)$ represent its majority color. Formally,
$\C(I) = \mathsf{R}$ if the interval $I$ has red majority, i.e. $\abs{R \cap I} > \abs{B \cap I}$. 
Similarly, set $\C(I) = \mathsf{B}$ if $\abs{R \cap I} < \abs{B \cap I}$. 
Without loss of generality, if $\abs{R \cap I} = \abs{B \cap I}$, i.e., there is no majority in $I$, we define $\C(I) = \mathsf{B}$. 

A point $i$ is {\it happy} in $\Pi$ if it is assigned to an interval matching its color, i.e., $\C(\pi_t) = x_i$, where $\pi_t \in \Pi$ is the interval containing $i$; 
otherwise $i$ is {\it unhappy} in $\Pi$. 
For an interval $I \subseteq [n]$, let $\unhappy{I, \Pi} \coloneqq \{ i \in I \mid i \text{ is unhappy in } \Pi \}$ denote the subset of unhappy points in $I$ in partition $\Pi$. Note that here $I$ can be any interval and is not necessarily a part in $\Pi$.
If $\Pi$ is fixed or clear from the context, we write $\unhappy{I} \coloneqq \unhappy{I, \Pi}$.

Intuitively, a partition $\Pi$ is {\it locally fair} if there is no large set of unhappy points which could form an allowable interval in which they would be the majority.
We make this concept more precise below: 

\begin{definition}
Given an input instance $(X, \sigma)$ and a parameter $\beta \in \left[\frac{1}{2}, 1\right]$, a \EMPH{$\beta$-deviating group} with respect to a partition $\Pi$ is an allowable interval $D$ with more than $\max \Set{\frac{|D|}{2}, \beta \sigma}$ unhappy monochromatic points, 
or equivalently, $\max \Set{\norm{\unhappy{D} \cap R}, \norm{\unhappy{D} \cap B}} > \max \Set{\frac{\norm{D}}{2}, \beta}.$ Note that we do not require the existence of a balanced partition with $D$ being one of its intervals.
\end{definition}

$D$ is called a $\beta$-deviating group because $D$ may deviate from $\Pi$ such that at least $\beta \sigma$ points that were unhappy in $\Pi$ become happy if $D$ was made a standalone part in another partition.
We sometimes omit $\beta$ and use the term ``deviating group'' when the context is clear.
Intuitively, $\beta$ controls how difficult it is for a set of unhappy points to deviate and form an interval in which they are happy. 
As $\beta$ grows, the set of unhappy points that may form a deviating group must grow larger with respect to the desired interval size $\sigma$; to deviate when $\beta = 1/2$, a set of more than $\sigma/2$ unhappy points must lie in an allowable interval in which they are the majority. 
At $\beta = 1$, the number of unhappy points required to form a deviating group increases to $\sigma$. See Figure~\ref{fig:intro} for an example.

\begin{definition}
Given $(X, \sigma)$, $\eps \in [0, 1/2]$, and $\beta \in [1/2, 1]$, we call a balanced partition $\Pi$ \EMPH{$(\eps, \beta)$-locally fair} if there is no $\beta$-deviating group with respect to $\Pi$.
\end{definition}

\begin{figure}[ht]
\centering
\includegraphics[width=.4\textwidth]{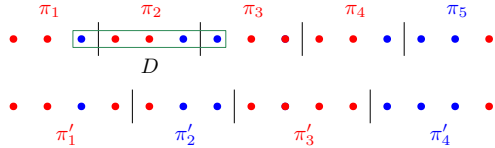}
\caption{An with $n = 15$ points, $\sigma = 4$, $\eps = 1/2$, and $\beta = 2/3$. The top partition $\{\pi_1, \dots, \pi_5\}$ admits a blue deviating group $D$, whereas the bottom partition $\{\pi'_1, \dots, \pi'_4 \}$ is $(1/2, 2/3)$-locally fair.} 
\label{fig:intro}
\end{figure}

\paragraph{Remark.}
While we focus on $\beta \in [1/2, 1]$, the largest possible range to consider is $\beta \in \left[(1-\eps)/2, (1+\eps) \right]$.
For applications of our model, we would expect the requirement on $\beta$ to be stronger than a simple majority, so that we bias towards solutions (partitions) which are presumably fair to the rest of the points. 

\section{Existence of Locally Fair Partitions}
\label{s:adversarial}

In this section, we present our results on the existence of locally fair partitions. 
We first give characterizations of parameters $\eps$ and $\beta$ for which a locally fair solution is guaranteed to exist. 
We show that for every $\eps \in [0, 1/2]$, there is a threshold $\bar{\beta}(\eps)$ such that if $\beta$ is above this threshold, a simple partitioning strategy into small intervals results in a $(\eps, \beta)$-fair partition. 

\begin{theorem} \label{thm:positive-exist}
Let $(X, \sigma)$ be an input instance as defined above. For any $\eps \in [0, 1/2]$, there is a value $\bar{\beta}(\eps)$ such that for any $\beta > \bar{\beta}(\eps)$, there exists an $(\eps, \beta)$-fair partition, where 
\[ \bar{\beta}(\eps) =  
\begin{cases} 
      \max \left\{ 1-\eps, \frac{1+3\eps}{2} \right\} + O(\delta) \text{ for } \eps \in [0, 1/3], \\
      \max \left\{ \frac{3(1-\eps)}{2},\  2\eps \right\} + O(\delta) \text{ for } \eps \in [1/3, 1/2],
  \end{cases}
\]
\noindent and $\delta \leq \frac{1}{\frac{n}{\sigma} - 1} + \frac{1}{\sigma}$.
\end{theorem}

\opt{aaai}
{
\paragraph{Proof sketch.} (See the full version for a detailed proof.)
We construct a partition $\Pi$ that uses as small as possible intervals (the length of each interval takes the form $(1 - \eps + \delta)\sigma$).
For $\eps \leq 1/3$, an allowable deviating group $D$ could intersect at most $3$ intervals. 
We show if $D$ intersects only $2$ intervals, a simple calculation shows there are less than $\beta \sigma$ unhappy points and such a $D$ cannot exist. 
On the other hand, if $D$ intersects $3$ intervals $\pi_i, \pi_{i+1}, \pi_{i+2}$, it must completely contain $\pi_{i+1}$.
We can bound the number of points $D$ can pull from $\abs{\pi_i \cup \pi_{i+2}}$ as $(1+\eps)\sigma - \abs{\pi_{i+1}}$, all of which could be unhappy.
Additionally, it must be $\abs{\unhappy{\pi_{i+1}}} \leq \abs{\pi_{i+1}/2}$.
Combining the two above observations shows the total number of unhappy points is less than $\beta \sigma$, and so no deviating group exists. 
A similar proof holds for $\eps \in [1/3, 1/2]$ by increasing the number of intervals the deviating group could intersect to $4$.
}

\opt{arxiv}{\begin{proof}
We first consider the case of $\eps \in [0, 1/3]$. We call a partition $\Pi$ of $[n]$ {\it almost-uniform} if there is some $q \in [1-\eps, 1+\eps]$ such that $\sigma' = q\sigma$ is an integer, and $\abs{\pi_t} \in \{ \sigma'-1, \sigma'\}$ for all $t = 1, \ldots, \abs{\Pi}$. 
Consider the following partition: construct $\lfloor \frac{n}{(1-\eps)\sigma} \rfloor$ intervals of $(1-\eps)\sigma$ points each. Then, uniformly distribute the remaining points (less than $(1-\eps)\sigma$ of them) across the intervals. Every interval then gets at most $\lceil \frac{(1-\eps)\sigma}{\lfloor \frac{n}{(1-\eps)\sigma} \rfloor} \rceil$ additional points. Therefore, the total number of points in the largest interval is given by
\begin{align*}
    \sigma' = q\sigma = (1-\eps)\sigma + \lceil \frac{(1-\eps)\sigma}{\lfloor \frac{n}{(1-\eps)\sigma} \rfloor} \rceil 
    \leq (1-\eps)\sigma + \frac{(1-\eps)\sigma}{\frac{n}{(1-\eps)\sigma}-1} +1 
    \leq (1-\eps)\sigma + \frac{(1-\eps)^2\sigma^2}{n-(1-\eps)\sigma}+1.
\end{align*}
Hence we have
\begin{align*}
    \delta = q - (1-\eps) \leq \frac{(1-\eps)^2\sigma}{n-(1-\eps)\sigma}+\frac{1}{\sigma}
    \leq \frac{\sigma}{n-\sigma} + \frac{1}{\sigma} = \frac{1}{\frac{n}{\sigma} - 1} + \frac{1}{\sigma}, 
\end{align*}
\noindent when $n \gg \sigma$. 

Hence take the almost-uniform partition $\Pi$ with this value of $q$.
\ 
Consider any potential deviating group $D$ of $\Pi$. 
First, suppose $D$ intersects at least four contiguous intervals of $\Pi$. Then it must contain at least two intervals in the middle, and therefore $\abs{D} > 2(1-\eps)\sigma \geq (1+\eps)\sigma$ for $\eps \in [0, 1/3]$, which violates the requirement that $D$ must be allowable. Hence, $D$ can intersect at most three (contiguous) intervals of $\Pi$.
\ 
If $D$ intersects one or two intervals of $\Pi$, then
\begin{align*}
    \abs{\unhappy{D, \Pi}} &\leq 2 \cdot \frac{\sigma'}{2} \leq \sigma' = q\sigma \leq \bar{\beta}(\eps)\sigma < \beta\sigma,
 \end{align*}
 \ 

\noindent which implies $D$ cannot be a deviating group itself. Hence, no deviating group intersects at most $2$ intervals. 
\ 
Finally, if $D$ intersects $3$ intervals, $\{\pi_t, \pi_{t+1}, \pi_{t+2} \}$, then it completely contains the middle interval $\pi_{t+1}$, which has size at most $\sigma'$.
Since at most half of the points in $\pi_{t+1}$ are unhappy, we have
$\abs{\unhappy{\pi_{t+1}}} \leq \frac{\sigma'}{2}.$
On the other hand, we have 
\begin{align*}
    \abs{D \cap (\pi_t \cup \pi_{t+2})} \leq (1+\eps)\sigma - \abs{\pi_{t+1}} \leq (1+\eps)\sigma - \Paren{\sigma'-1}.
\end{align*}  
Thus we have
\begin{align*}
    \abs{\unhappy{D, \Pi}} = \sum_{j=t}^{t+2} \abs{\unhappy{D \cap \pi_j}} &\leq \abs{D \cap (\pi_i \cup \pi_{t+2})} + \abs{\unhappy{\pi_{t+1}}} \tag*{}\\
    &\leq (1+\eps)\sigma - \sigma' + 1 + \frac{\sigma'}{2} \tag*{}\\
    &\leq \Paren{1+\eps-\frac{q}{2}+\frac{1}{\sigma}}\sigma\\
    &\leq \Paren{\frac{1+3\eps}{2}-\frac{\delta}{2}+\frac{1}{\sigma}}\sigma\\
    &\leq \bar{\beta}(\eps)\sigma < \beta \sigma, \tag*{}
\end{align*}
\ 

\noindent which again implies $D$ cannot be a deviating group.
Hence $D$ does not exist, and the proposed almost-uniform partition $\Pi$ is $(\eps, \beta)$-locally fair for $\eps \in [0, 1/3]$. The proof for $\eps \in [1/3, 1/2]$ follows analogously by increasing the largest number of intervals the deviating group can intersect by one, from three to four.
\end{proof}}

\opt{arxiv}{\paragraph{Remark.} 
When $n = c(1-\eps)\sigma$ for some integer $c$, the partition proposed in the proof of Theorem~\ref{thm:positive-exist} becomes a uniform partition with interval size $(1-\eps)\sigma$, and thus $\delta = 0$. 
Otherwise, we try to create an almost-uniform partition such that the intervals are as small as possible. 
For Theorem~\ref{thm:positive-exist} to give a $\bar{\beta}(\eps)$ bound below $1$, we need $\delta$ to be roughly $O(\eps)$, or $\eps > \sigma/n$. 
If $\sigma/n = c$ for a constant $c$, the desired interval size is $\Theta(n)$, and the number of intervals in a balanced partition will only be $O(1)$. In this setting, it is likely that any fair partitioning scheme must take into account the coloring of the input.}

Theorem~\ref{thm:positive-exist} can be extended for any 
$\eps \in [0, 1]$. Following the same proof approach gives the general form of 
$\bar{\beta}(\eps) = \max \left\{ \frac{t(1-\eps)}{2}, \frac{(3-t) + (t+1)\eps}{2}\right\} + O(t \sigma)$, where $t$ is an integer such that $\eps \in \left[\frac{t-2}{t}, \frac{t-1}{t+1}\right]$, or in other words, $(t+1)$ is the largest number of intervals a deviating group can intersect.


\smallskip Next, we show that for smaller values of $\beta$, a locally fair partition may not exist. 

\begin{theorem}\label{thm:neg-exist}
Let $\eps \in [0, 1/2)$ and $\beta \in [1/2, 1 - \eps)$.
For any $\sigma \geq 1$,
there exists an input instance $(X, \sigma)$ with $|X| = O\Paren{\frac{\beta \sigma}{1 - \eps - \beta}}$ for which no $(\eps, \beta)$-locally fair partition exists.
\end{theorem}


\opt{aaai,arxiv}
{
\begin{proof}
We construct an instance for which there is always a deviating group. For simplicity, assume $\beta \sigma$ is an integer; we will relax this assumption later. Specifying the input using the runs of monochromatic intervals, let $X = R_1, B_1, R_2, B_2, \dots, R_\eta, B_\eta$ for $\eta = \ceil{\frac{n}{2\beta\sigma}}$. Set $\abs{R_j} = \beta \sigma$ for $j = 1, \dots, \eta$ and $\abs{B_j} = \beta \sigma$ for $j = 1, \dots, \eta - 1$, and 
$\abs{B_\eta} = n - (2\eta - 1) \beta \sigma \leq \beta\sigma.$

In any fair partition $\Pi$, each $R_j$ (resp. $B_j$) must intersect a red (resp. blue) interval of $\Pi$;
if there exists an entire $R_j$ (resp. $B_j$) contained in a blue (resp. red) interval of $\Pi$, $R_j$ (resp. $B_j$) could form a deviating group with its own $ \beta \sigma$ unhappy points, and $ (1 - \eps - \beta) \sigma$ points from a neighboring $B_j$ (resp. $R_j$). 
Since $\beta \geq 1/2$, these unhappy red (resp. blue) points are the majority in the deviating group. As a consequence, in any fair partition $\Pi$, there exists no red (resp. blue) interval $\pi_t$ that intersects multiple monochromatic red intervals $R_j$, $R_{j'}$ (resp. blue intervals $B_j$, $B_{j'}$) in $X$.

Suppose there exists a fair partition $\Pi$ for $(X, \sigma)$, and let
let $\pi_1$ be a red interval of $\Pi$ that intersects $R_1$.
Since $\beta < 1 - \eps$, $\pi_{1}$ must include points from a neighboring blue monochromatic interval; without loss of generality, let it intersect with $B_1$.\footnote{We assume the input lies on a circle. Since $\pi_1$ must intersect at least one of the two blue monochromatic intervals neighboring $R_1$, we can order the input so that the intersected interval is $B_1$.} Then $\pi_{1}$ cannot intersect $R_2$; otherwise $B_1$ would deviate.
Therefore, $\pi_1$ includes at most $\beta\sigma$ points from $R_1$ and some points from $B_1$.
Now, consider interval $\pi_{2}$ of $\Pi$ which is blue and intersects $B_1$ (but not $B_2$). 
By size constraints, $\norm{\pi_{2}} \geq (1-\eps)$, and
$\norm{\pi_{2} \cap B_1} < \beta$, since $\pi_{1} \cap B_1 \neq \emptyset$. This implies $\norm{\pi_{2} \cap R_2} > (1 - \eps - \beta)$. Next, interval $\pi_3$ is red-majority and intersects $R_2$ (but not $R_3$); moreover, we have $\norm{\pi_{3} \cap R_2} \leq \norm{R_2} - \norm{\pi_{2} \cap R_2} < \beta - (1 - \eps - \beta)$. 
This then implies $\norm{\pi_{3} \cap B_2} \geq \norm{\pi_{3}} - \norm{\pi_{3} \cap R_2} > (1-\eps) - \Paren{\beta - (1 - \eps - \beta)} = 2(1 - \eps - \beta)$. 

Continuing this argument, it can be shown that (i) every $R_j$ must intersect a red interval in $\Pi$ that also intersects with $B_{j}$; and (ii) every $B_j$ must intersect a blue interval in $\Pi$ that also intersects with $R_{j+1}$. Combined with the fact that each red- (resp. blue-) majority $\pi_t$ cannot intersect multiple monochromatic red (resp. blue) intervals of $X$, for every $j$ it must hold that:
\begin{itemize}
    \item $\pi_{2j-1}$ is red-majority and intersects $R_j$;
    \item $\pi_{2j}$ is blue-majority and intersects $B_j$;
    \item $\norm{\pi_{2j-1} \cap R_j} < \beta - (2j-3)\cdot(1-\eps-\beta) $;
    \item $\norm{\pi_{2j-1} \cap B_j} > (2j-2)\cdot(1 - \eps - \beta)$;
    \item $\norm{\pi_{2j} \cap B_j} < \beta - (2j-2)\cdot(1-\eps-\beta)$;
    \item $\norm{\pi_{2j} \cap R_{j+1}} > (2j-1)\cdot(1 - \eps - \beta)$.
\end{itemize}


Denote $j' = \ceil{\frac{3-3\eps-2\beta}{4(1-\eps-\beta)}}$, and assume $\eta \geq j'$. By the above, $\pi_{2j'}$ is blue-majority, and we have $\norm{\pi_{2j'} \cap B_{j'}} < \beta - (2{j'}-2) \cdot (1-\eps-\beta) < \frac{1-\eps}{2}$, which implies $\pi_{2j'}$ cannot be blue-majority, a contradiction. 
In other words, there are not enough points in $\pi_{2j'}$ to create a majority matching its color.
Since $\eta = \ceil{\frac{n}{2\beta\sigma}}$, the above holds for $\frac{n}{2\beta\sigma} > \frac{3-3\eps-2\beta}{4(1-\eps-\beta)}$,
or $n > \frac{(3-3\eps-2\beta)\beta\sigma}{2(1-\eps-\beta)} = 
O\Paren{\frac{\beta \sigma}{1 - \eps - \beta}}$. 
Finally, if $\beta \sigma$ is not an integer, let $\beta' = \frac{\ceil{\beta \sigma}}{\sigma}$. Then the above argument still holds for $n > \frac{(3-3\eps-2\beta')\beta'\sigma}{2(1-\eps-\beta')}$.
\end{proof}
}

\begin{figure}[t]
\centering
\begin{tikzpicture}
	\node[rectangle,draw,fill = red!20, minimum width = 1cm, 
    minimum height = 0.6cm] (r1) at (0,0) {};
    \node[rectangle,draw,fill = blue!20, minimum width = 1cm, 
    minimum height = 0.6cm] (b1) at (1,0) {};
    \node[rectangle,draw,fill = red!20, minimum width = 1cm, 
    minimum height = 0.6cm] (r2) at (2,0) {};
    \node[rectangle,draw,fill = blue!20, minimum width = 1cm, 
    minimum height = 0.6cm] (b2) at (3,0) {};
    \node[rectangle,draw,fill = red!20, minimum width = 1cm, 
    minimum height = 0.6cm] (r3) at (4,0) {};

    
    \draw[|-|, red]([yshift=-3mm]r1.south west) -- node[below] {$\pi_1$} ([yshift=-3mm,xshift=-2mm]b1.south);
    
     \draw[|-|, blue]([yshift=-3mm,xshift=-2mm]b1.south) -- node[below] {$\pi_2$} ([yshift=-3mm,xshift=-1mm]r2.south);
     
    \draw[|-|, red]([yshift=-3mm,xshift=-1mm]r2.south) -- node[below] {$\pi_3$} ([yshift=-3mm,xshift=0mm]b2.south);
      
     \draw[|-|, red]([yshift=-3mm,xshift=-0mm]b2.south) -- node[below] {$\pi_4$} ([yshift=-3mm,xshift=0mm]r3.south);
    
    \node[] (c) at (5,0) {$\cdots$};
    
    \draw[<->]([yshift=3mm]b2.north west) -- node[above] {$B_\textsl{\textsc{d}}$} ([yshift=3mm]b2.north east);
    
    \draw[<->]([yshift=3mm]r1.north west) -- node[above] {$\frac{5}{8}\sigma$} ([yshift=3mm]r1.north east);
\end{tikzpicture}
\caption{
An instance so that each monochromatic interval has length $5\sigma/8$, which does not admit a $(\frac{1}{4}, \frac{5}{8})$-fair partition. 
For example, partition $\Pi$ is made of intervals of length $(1-\eps) \sigma = 3\sigma/4 $; however,
$B_\textsl{\textsc{d}}$ forms a deviating group by pulling in points from neighboring intervals.} 
\label{fig:bad}
\end{figure}
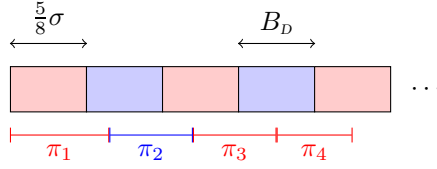

See Figure~\ref{fig:bad} for an example of the construction. 
For $\eps \in [0, 1/5]$, Theorem~\ref{thm:positive-exist} and Theorem~\ref{thm:neg-exist} provide an almost sharp threshold: if $\beta > 1-\eps+O(\delta)$ (for $\delta$ defined in Theorem~\ref{thm:positive-exist}), a locally fair partition always exists, but for $\beta < 1-\eps$ there are instances that do not admit fair partitions. 
In fact, if we enforce the deviating group to have exactly $\sigma$ points, Theorems~\ref{thm:positive-exist} and~\ref{thm:neg-exist} become almost tight for all $\eps \in [0, 1/3]$. 
We next extend Theorem~\ref{thm:neg-exist} so that a single instance $X$ has no locally fair partition for a wide range of $\sigma$ values. 

\begin{corollary}
Let $\eps \in [0, 1/2)$, $\beta \in [1/2, 1-\eps)$, and let 
$\mathcal{S} = \{ \sigma_1, \sigma_2, \dots, \sigma_M\}$ be the set of desired interval sizes. 
If $\frac{n}{M\sigma_m} > \ceil{\frac{1}{1-\eps-\beta}}$ holds for all $m \in [M]$, there exists an input $X$ such that (i) $\abs{X} = n$; (ii) for all $m = 1, \dots, M$, the instance $(X, \sigma_m)$ has no $(\eps, \beta)$-locally fair partition. 
\end{corollary}

\opt{arxiv}{\begin{proof}
We construct $X$ as follows. 
Let $X_1, \dots, X_M \subseteq X$ be subintervals of $X$, each of size $n/M$. 
For each part $X_m$, we apply Theorem~\ref{thm:neg-exist} with $(X_m, \sigma_m)$. 
For each $\sigma_m$, a deviating group exists in $X_m$. 
\end{proof}

\paragraph{Remark.} (i) In light of Theorem~\ref{thm:neg-exist}, we observe that if we are given a periodic instance, it is always possible to define a $\tilde{\sigma}$ for which the instance has a locally fair solution, even if the instance had no locally fair solution for the given $\sigma$. Specifically, $\tilde{\sigma} = \Theta(\sigma)$.
In the next section, we define an approximate version of the problem, in which we allow some fraction of points to lie outside of allowable regions. 
We focus on ``periodic-like'' instances, when the input consists of monochromatic intervals of size $\Omega(\sigma)$.

(ii) In some applications, it may be desirable to ensure there is a fixed number of intervals in a partition. 
Given a desired number of intervals $k$, the negative results of Theorem~\ref{thm:neg-exist} extend to this setting, i.e., we can construct periodic instances similar to those in Theorem~\ref{thm:neg-exist} in which no locally fair solutions exist, for even larger ranges of $(\eps, \beta)$ parameters.}

\section{Clustered Instances}
\label{s:periodic}

As manifested in the previous section, under many specific parameters $\eps, \beta$, there exist adversarial input instances $(X, \sigma)$ that rule out the existence of any locally fair partition. 
However, such negative instances often seem artificial, and are not robust to perturbation. 
In this section, we turn our attention to a category of interesting inputs, when points are ``clustered'' into large monochromatic intervals. 
Such instances arise in applications in which we expect points of the same color to gather together. 
We show that fair partitions exist when the input instance is comprised of large monochromatic intervals, while incurring a small approximation on the balancedness of the fair partition. 

For a constant $\alpha \in [0,1)$, we say a partition $\Pi$ of $X$ is \EMPH{$\alpha$-balanced} if the union of all its allowable intervals make up at least a $(1-\alpha)$-fraction of the total input.
Formally, let $\Tilde{\Pi} \coloneqq \{ \pi_t \in \Pi \mid \abs{\pi_t} \in [(1-\eps)\sigma, (1+\eps)\sigma]\}$ be the set of allowable intervals in $\Pi$. 
Then $\Pi$ is $\alpha$-balanced if $\abs{\bigcup_{\pi_t} \{ \pi_t \in \Tilde{\Pi} \}} \geq (1-\alpha) n$.

In fact, in this section our results hold for any $\beta \in [1/2, 1]$, so we omit $\beta$ as a parameter, and instead refer to a fair partition as $\eps$-locally fair. 


First, we show that if the size of each monochromatic interval in $X$ is at least $2\sigma$, we can compute a fair partition by letting allowable intervals not contain a small fraction of the population. 


\begin{theorem}\label{result:partitioning_constant_intervals}
Given instance $\Paren{X = R_1, B_1, R_2, \dots, \sigma}$ with $\norm{R_j}, \norm{B_j} \geq 2,$ for all $j$ and parameter $\eps \in [0, 1/2]$, 
there is a $\Paren{\frac{(1-\eps)\sigma}{n}}$-balanced, $\eps$-locally fair partition $\Pi$.
\end{theorem}

\begin{proof}
We assume that the first monochromatic interval $R_1$ is red and the last monochromatic interval $B_\eta$ is blue. 
The other cases follow analogously. 
Consider an arbitrary maximal monochromatic red interval $R_j$ with measure $\norm{R_j} \geq 2$. 
We divide $R_j$ into allowable intervals such that the residual interval $R'_j$ (points of $R_j$ not assigned to an allowable interval) is as small as possible. 
Note that $\norm{R'_j} \in [0, 1-\eps)$. 
We assign the points of $R'_j$ to a size $\lceil (1-\eps)\sigma \rceil $ interval using $\lceil \Paren{(1-\eps) - \norm{R'_j}}\sigma \rceil$ points from the next interval, $B_j$.
Partition the entire instance in this manner and let $\Pi$ be the resulting partition. 
All intervals of $\Pi$ are allowable except for the residual interval $B'_{\eta}$ from the last monochromatic interval $B_{\eta}$. 
Since $|B'_{\eta}| < (1-\eps) \sigma$, $\Pi$
is $\Paren{\frac{(1-\eps)\sigma}{n}}$-balanced.

We next show $\Pi$ is locally fair regardless of the value of $\eps$. 
Suppose there exists a deviating group $D$. 
Without loss of generality, assume $D$ is red-majority and $D$ intersects two consecutive red intervals $R_j$ and $R_{j+1}$ of $X$. 
Then we have $B_{j+1} \subset D$. 
Since $\norm{B_{j+1}} \geq 2$, we have $\norm{D} > 2 > (1+\eps)$, a violation of the size constraint. 
Hence, $D$ can only intersect one monochromatic red interval $R_j$ for some $j$.

Define $R''_j \subseteq R_j$ (resp. $R'_j \subseteq R_j$) to be the (not necessarily non-empty) set of red points assigned to the same interval, denoted by $\pi''$ (resp. $\pi'$) as a subset $B'_j$ of $B_j$ (resp. $B'_{j+1}$ of $B_{j+1})$ (See Figure~\ref{fig:thm6}).
If $\pi'$ is blue majority, then $\norm{R''_j} < \frac{1-\eps}{2}$, by construction.
Similarly, if $\pi'$ is blue majority, then $\norm{R'_j} < \frac{1-\eps}{2}$. 
Assume $D$ intersects both $R'_j$ and $R''_j$. 
Then both $\pi'$ and $\pi''$ need to be blue-majority, and we have $R \setminus \Paren{R'_j \cup R''_j} \subset D$. 
But this implies $\norm{R \setminus \Paren{R'_j \cup R''_j}} > 2 - \frac{1-\eps}{2} - \frac{1-\eps}{2} = 1+\eps$ and thus $\norm{D} > 1+\eps$, a contradiction. 

Hence, $D$ can only intersect either $R'_j$ or $R''_j$. 
In both cases, $\norm{\unhappy{D}} < \frac{1-\eps}{2}$, implying $D$ does not have sufficient points to form a deviating group. 
Therefore, no deviating group exists in $\Pi$. 
\end{proof}

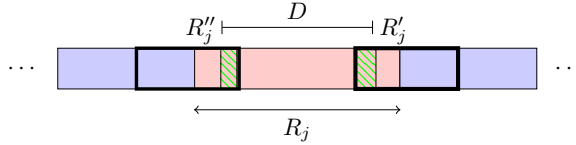
\begin{figure}[t]
\centering
\scalebox{0.9}{
\begin{tikzpicture}
	\node[rectangle,draw,fill = blue!20, minimum width = 2cm, 
    minimum height = 0.6cm] (bi) at (0,0) {};
    \node[rectangle,draw,fill = red!20, minimum width = 3cm, 
    minimum height = 0.6cm] (ri) at (2.5,0) {};
    \node[rectangle,draw,fill = blue!20, minimum width = 2cm, 
    minimum height = 0.6cm] (ri1) at (5,0) {};
    
     \node[draw,  fill=red!30, minimum width = 0.1cm,
    minimum height = 0.6cm, pattern=north west lines, pattern color=green] (d1) at (1.5,0) {};
     \node[draw,  fill=red!30, minimum width = 0.3cm,
    minimum height = 0.6cm, pattern=north west lines, pattern color=green] (d2) at (3.5,0) {};
    
     \node[rectangle,draw, minimum width = 1.5cm, 
    minimum height = 0.6cm, line width=.5mm ] (g1) at (.9,0) {};
     \node[rectangle,draw, minimum width = 1.5cm, 
    minimum height = 0.6cm, line width=.6mm ] (g2) at (4.1,0) {};
    
    \node[] (res1) at ([yshift=2.5mm,xshift=1mm]ri.north west) {$R_j''$};
     \node[] (res2) at ([yshift=2.5mm,xshift=-1mm]ri.north east) {$R_j'$};
    
    \node[] (g2) at (-1.5,0) {$\cdots$};
    \node[] (g2) at (6.5,0) {$\cdots$};

    \draw[<->]([yshift=-3mm]ri.south west) -- node[below] {$R_j$} ([yshift=-3mm]ri.south east);
    
    \draw[|-|]([yshift=3mm, xshift=4mm]ri.north west) -- node[above] {$D$} ([yshift=3mm,xshift=-4mm]ri.north east);
\end{tikzpicture}
}
\caption{No deviating group $D$ intersects $R''_j$ and $R'_j$.}
\label{fig:thm6}
\end{figure}

In fact, we can use the same partitioning strategy to find balanced fair partitions (i.e., every point belongs to an allowable interval $\pi_t$) if each monochromatic interval of an instance has size at least $\ceil{\frac{(1-\eps)^2}{2\eps}}$. 

\begin{corollary}
\label{cor:large-eps}
For an instance $(X, \sigma)$ and parameter $\eps$, such that for all $j$: $\norm{R_j}, \norm{B_j} \geq \ceil{\frac{(1-\eps)^2}{2\eps}},$
there is always a balanced $\eps$-locally fair partition $\Pi$ of $X$.
\end{corollary}

\opt{arxiv}{\begin{proof}
Observe that the residual interval for each $R_j, B_j$ defined in Theorem~\ref{result:partitioning_constant_intervals} has length zero. 
Accordingly, by uniformly partitioning each $R_j, B_j$ into equally-sized monochromatic intervals, we have a locally fair partition of $X$ in which no points are unhappy.
\end{proof}}

\opt{arxiv}{\paragraph{Remark.} For all $\eps \geq 3 - 2\sqrt{2} \approx .17$, Corollary~\ref{cor:large-eps} requires monochromatic intervals to have size at most that of Theorem~\ref{result:partitioning_constant_intervals}, but achieves exact balancedness. 
However, the threshold $\ceil{\frac{(1-\eps)^2}{2\eps}}$ is non-increasing in $\eps$: when $\eps \rightarrow 0$, $\ceil{\frac{(1-\eps)^2}{2\eps}} \rightarrow \infty$, while when $\eps \geq 1/2$, $\ceil{\frac{(1-\eps)^2}{2\eps}} = 1$.
Thus, for small values of $\eps$, it is preferable to apply Theorem~\ref{result:partitioning_constant_intervals}.}

Next, we relax the requirement that \textit{every} monochromatic interval is long and consider ``mostly'' clustered instances. 
Specifically, assume that at most $\gamma n$ points of $X$ lie in monochromatic intervals smaller than size $2\sigma$. 
Applying Theorem~\ref{result:partitioning_constant_intervals} to this setting, we can construct an $\alpha$-balanced fair partition with $\alpha$ depending on $\gamma$. 

\begin{theorem}
\label{lem:good-bad-partitioning} 
Given $(X, \sigma)$ and parameter $\eps \in [0, 1/2]$, where $X = R_1, B_1, \dots, R_\eta, B_\eta$, 
let $Y$ denote the set of monochromatic intervals of $X$ of length smaller than $2\sigma$: 
$Y \coloneqq \{ I \in \{R_1, \ldots, R_\eta, B_1, \ldots, B_\eta \} \mid \norm{I} < 2 \}.$

If $\abs{\bigcup_{I \in Y} I} \leq \gamma n$, 
then there is a $\Paren{\frac{(1-\eps)\sigma}{n} + \gamma}$-balanced, $\eps$-locally fair partition.
\end{theorem}



\opt{arxiv}{\begin{proof}
We partition $X$ into two parts: the {\it bad} part $Y$ and the {\it good} part $Z = X \setminus Y$. 
$X$ can be regarded as an alternating sequence of $Z_i$s and $Y_j$s.
 
Apply Theorem~\ref{result:partitioning_constant_intervals} to each good part $Z_i$. 
For each $Z_i$, the residual points $Z'_i$ of the last monochromatic interval in $Z_i$ satisfies $\norm{Z'_i} < (1-\eps)$. 
Then we consider two cases: 
\smallskip
\begin{itemize}[itemsep=0pt, topsep=0pt]
    \item $\norm{Z'_i} + \norm{Y_i} \geq 1-\eps$. We define an interval $I_i$ containing $Z'_i$ and $(1-\eps -\norm{Z'_p})\sigma$ points of $Y_i$, and add $I_i$ to the partition.
    For each remaining monochromatic interval $R_j$ (or $B_j$) in $Y_i$, we create a (not necessarily allowable) standalone interval in our partition, so that all points are happy and do not contribute to any potential deviating groups.
    \item $\norm{Z'_i} + \norm{Y_i} < 1-\eps$. In this case, we apply the partitioning scheme in the proof of Theorem~\ref{result:partitioning_constant_intervals} on $Z_i, Y_i, Z_{i+1}$, namely, there is an interval $\pi_i$ containing $Z'_i, Y_i$ and $\lceil \Paren{1 - \eps - \abs{Z_i} - \abs{Y_i}}\sigma \rceil$ points of the first monochromatic interval of $Z_{i+1}$. 
\end{itemize}
\smallskip
\noindent In both cases, no deviating group exists since the interval containing $Z'_i$ and points of $Y_i$ are adjacent to intervals with no unhappy points, although there may be a total of $\gamma n$ points lying in non-allowable intervals $\pi_i \in \Pi$ with $\abs{\pi_i} < (1-\eps)\sigma$. 
\ 
Additionally, for the last good part $Z_i$, there may be an additional $\Paren{\frac{(1-\eps)\sigma}{n}}$ points lying in non-allowable intervals, as a direct consequence of applying Theorem~\ref{result:partitioning_constant_intervals}. (Note that its counterparts in $Z_1, \ldots, Z_{i-1}$ are handled in the above process.)
Thus, there is a $\Paren{\frac{(1-\eps)\sigma}{n} + \gamma}$-balanced, $\eps$-locally fair partition.
\end{proof}}

For all $\eps$, $\alpha$ improves (i.e., decreases) as $\gamma$ decreases, as more of the input lies in larger monochromatic intervals. 
Similar to Corollary~\ref{cor:large-eps}, if $\eps \geq 3 - 2\sqrt{2}$, we can prove the above process gives a $\gamma$-balanced, $\eps$-locally fair partition. 


\opt{arxiv}{\paragraph{Remark.} Here we defined the approximation concept in terms of the balancedness of fair partitions. If we want to ensure the partition is strictly balanced, we can instead define the approximation concept in terms of the total number of points in the union of deviating groups, i.e., we allow a maximum of $\alpha$-fraction of the points to deviate (participate in some deviating group). 
Still assuming that at most $\gamma n$ points in the input lie in monochromatic intervals smaller than size $2\sigma$, if $\eps \geq 3 - 2\sqrt{2}$, we can get an exactly balanced partition which is $\Paren{\frac{\gamma}{2} + \frac{1-\eps}{4}}$-approximately fair.
For smaller $\eps$, we still require a non-zero $\alpha$, but the $\gamma$ parameter could be moved into the fairness approximation.}

\section{Partitioning Algorithm}
\label{s:fair-dp}

Finally, we shift our focus to following algorithmic question: {\it Given an instance $(X, \sigma)$ and parameters $\eps \in [0, 1/2], \beta \in [1/2, 1]$, does a $(\eps, \beta)$-locally fair solution exist for $(X, \sigma)$?}

In this section, we focus on a fixed input instance, so throughout this section, treat $X, \sigma, \eps$, and $\beta$ as fixed, and describe an algorithm that determines whether an $(\eps, \beta)$-locally fair balanced partition, possibly with additional constraints, exists, for an interval $I \subseteq [n]$ of the input, where $|X| = n$.

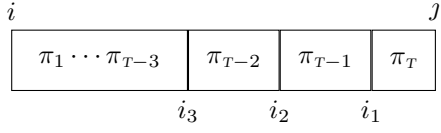
\begin{figure}[t]
\centering
\begin{tikzpicture}
\normalsize
\matrix (m) [matrix of nodes,
             nodes={draw, minimum size=8mm,anchor=center},
             nodes in empty cells, minimum height = .5cm,]
{
  $\ \ \pi_1 \cdots  \pi_{\textsl{\textsc{t}}-3} \ \ $ & $\ \pi_{\textsl{\textsc{t}}-2}\ $ & $\ \pi_{\textsl{\textsc{t}}-1}\ $ & $\ \pi_{\textsl{\textsc{t}}}\ $     \\
};
\node [align=left, anchor=south] at (m-1-1.north west) (i) {$i$};
\node [align=right, anchor=south] at (m-1-4.north east) (j) {$j$};
\node [align=right, anchor=north] at (m-1-2.south west) (3) {$i_3$};
\node [align=right, anchor=north] at (m-1-3.south west) (2) {$i_2$};
\node [align=right, anchor=north] at (m-1-4.south west) (1) {$i_1$};
\end{tikzpicture}
\caption{$\Pi$ is a partition of interval $(i, j]$, where $i_1, i_2$, and $i_3$ are the last three boundaries.}
\label{fig:dp}
\end{figure}

We now define the recursive subproblems. 
For any interval $I = (i,j] \in [n]$ and 
for $i \leq i_3 \leq i_2 \leq i_1 < j$, define $\mathsf{LF}(I, i_1, i_2, i_3) = \mathsf{True}$ if and only if there exists at least one fair partition $\Pi = \{\pi_1, \ldots, \pi_T\}$ of $I$ that satisfies the following conditions:
\begin{footnotesize}
\begin{align*}
    &i_1 = i_2 = i_3 = i, \quad \Pi = \{(i,j]\}, \quad \text{for } T=1;\\
    &i_1 > i_2 = i_3 = i, \quad \Pi = \{(i,i_1],(i_1,j]\}, \quad \text{for } T=2;\\ 
    &i_1 > i_2 > i_3 = i, \quad \Pi = \{(i,i_2],(i_2,i_1],(i_1,j]\}, \quad \text{for } T=3;\\ 
    &i_1 > i_2 > i_3 > i, \quad \Pi = \{\pi_1, \ldots, (i_3,i_2],(i_2,i_1],(i_1,j]\}, \, \text{o/w.}
\end{align*}
\end{footnotesize}
\noindent In other words, $i_1, i_2$, and $i_3$ are the last three ``interval boundaries'' in $\Pi$; see Figure~\ref{fig:dp}. 




We first define the base cases of our algorithm. Consider every interval $I = (i,j]$ that is allowable, i.e., $(1-\eps)\sigma \leq \abs{I} \leq (1+\eps)\sigma$. Without loss of generality, let 
$\chi(I) = \mathsf{B}$. We consider letting $\Pi = \{(i,j]\}$ be the trivial partition for $I$. This is locally fair if and only if no deviating group can form within $I$, i.e., there exists no $(i', j'] \subset 
(i,j]$ such that (i) $(1-\eps)\sigma \leq (j'-i')$ (so that $(i',j']$ is allowable), and (ii) $\abs{\unhappy{(i',j'] \cap R}} \geq \max\{ \frac{\beta \sigma}{2}, \frac{j'-i'}{2}\}$ (so that $(i', j']$ is deviating). If the above holds, we have $\mathsf{LF}(I, i, i, i) = \mathsf{True}$, and $\mathsf{LF}(I, i_1, i_2, i_3) = \mathsf{False}$ for any other values of $i_1, i_2,$ and $i_3$.


Intuitively, if an interval $I = (i,j] \subseteq [n]$ has a fair partition, either $I$ is a standalone fair allowable interval (i.e., $\mathsf{LF}((i,j], i, i, i) = \mathsf{True}$), or there must exist a point $j' \in [j - (1+\eps)\sigma, j-(1-\eps)\sigma)$ such that 
\begin{description}
\item[(i)] there is a fair partition for $(i,j']$;
\item[(ii)] $(j',j]$ is a fair unpartitioned standalone interval; 
\item[(iii)] no deviating group is formed by the last 3 intervals of fair partition of $(i,j']$ and the interval $(j', j]$.
\end{description}

\opt{arxiv}{
To see this, simply observe that any fair partition of $(i,j]$ that contains at least two intervals has at least one interval boundary $j'$ between $(j - (1+\eps)\sigma)$ and $(j - (1-\eps)\sigma)$ (so that the last interval $(j',j]$ is allowable).
There is no deviating group formed within interval $(i, j']$, $(j', j]$ or straddling $j'$ (such a deviating group cannot span more than three intervals of a balanced partition of $(i, j']$).
}
Accordingly, we have the following lemma:


\begin{lemma}
\label{lem:dp}
$\mathsf{LF}((i,j], i_1, i_2, i_3) = \mathsf{True}$
if and only if there exists a point $i_4 \in [\max\{ i, i_3 - (1+\eps)\sigma\}, i_1]$ such that
\begin{itemize}
    \item $\mathsf{LF}((i_1, j], i_1, i_1, i_1) = \mathsf{True}$, and
    \item $\mathsf{LF}((i, i_1], i_2, i_3, i_4) = \mathsf{True}$, and
    \item There is no deviating group forming within $(i_4,j]$ with the partition $\{(i_4, i_3], (i_3, i_2], (i_2, i_1], (i_1, j] \}$.
\end{itemize}
\end{lemma}


\opt{arxiv}
{
\begin{proof}
Recall that $\eps \leq \frac{1}{2}$. Thus, every allowable interval, as well as a potential deviating group, has size between $\frac{\sigma}{2}$ and $\frac{3\sigma}{2}$. Suppose there exists a deviating group $D$ for a partition $\Pi$ of $(i,j]$ intersecting more than four contiguous intervals in $\Pi$. Then $D$ must strictly contain at least three intervals $\pi_i, \pi_{i+1}, \pi_{i+2}$, and thus we have $\abs{D} > 3 \cdot \frac{\sigma}{2}$, a contradiction. Hence $D$ can intersect at most four contiguous intervals in any partition $\Pi$ for $X$.
\ 
Suppose $\mathsf{LF}((i_1, j], i_1, i_1, i_1) = \mathsf{LF}((i, i_1], i_2, i_3, i_4) 
= \mathsf{True}$ 
for some $i_4$. This means (i) interval $(i_1, j]$ is fair as a standalone interval; (ii) there is a fair partition $\Pi'$ of interval $(i, i_1]$, with the last three interval boundaries being $i_2$, $i_3$, and $i_4$. Concatentate $\Pi'$ with $(i_1, j]$ as $\Pi$. By the discussion above, any deviating group $D$ for $\Pi$ intersects at most four contiguous intervals in $\Pi$.
\ 
Suppose a such $D$ intersects $(i,i_4]$. Then it must not intersect $(i_1, j]$ as a consequence, which implies $D$ is also a deviating group for $\Pi'$, a contradiction. Hence $D$ must be contained in $(i_4, j]$. However, 
by the third criteria, there is no deviating group in $(i_4,j]$ either. Therefore, no deviating group exists in $D$, which means $\Pi$ is a fair partition for $(i,j]$.
\end{proof}
}

\opt{arxiv}
{
In other words, for $i_1, i_2$, and $i_3$ to be the last three interval boundaries in a fair partition for $(i,j]$, the last interval $(i_1, j]$ must be fair itself, and there must exist a point $i_4$ as the next (fourth last) interval boundary. Note that it is possible that $i_4 = i_3$ in the degenerate case when $(i,j]$ is partitioned into less than 4 intervals. Accordingly, we know that $(i,i_1]$ must also be fair with last three interval boundaries being $i_2, i_3$, and $i_4$, as well as $(i_4,j]$ must be fair with the interval boundaries being $i_1, i_2$, and $i_3$.

For $\eps \leq 1/3$, as $(i_4, i_3]$ cannot be further separated into multiple allowable intervals, Lemma~\ref{lem:dp} can be simplified as follows: 

\begin{corollary}
\label{cor:small-eps-dp}
For $\eps \leq 1/3$, $\mathsf{LF}((i,j], i_1, i_2, i_3) = \mathsf{True}$
if and only if there exists a point $i_4 \in [\max\{ i, i_3 - (1+\eps)\sigma\}, i_1]$ such that
\begin{itemize}
    \item $\mathsf{LF}((i_1, j], i_1, i_1, i_1) = \mathsf{True}$, and
    \item $\mathsf{LF}((i, i_1], i_2, i_3, i_4) = \mathsf{True}$, and
    \item $\mathsf{LF}((i_4, j], i_1, i_2, i_3) = \mathsf{True}$.
\end{itemize}
\end{corollary}

\begin{proof}
The proof follows from Lemma \ref{lem:dp}. Since $\mathsf{LF}((i_4, j], i_1, i_2, i_3) = \mathsf{True}$, there must exist a fair partition for $(i_4,j]$ with the last three intervals being $\{(i_3, i_2], (i_2, i_1], (i_1, j] \}$. By the fact that $\mathsf{LF}((i, i_1], i_2, i_3, i_4) = \mathsf{True}$, we know $(i_4, i_3]$ is an allowable region. For $\eps \in [0, 1/3)$, this implies $\abs{(i_4, i_3]} \leq (1+\eps)\sigma < \frac{2\sigma}{3} < \frac{2(1-\eps)\sigma}{3}$, i.e., there is no other way to partition $(i_4, i_3]$ in a balanced fashion other than making it a standalone part. Hence, the partition $\{(i_4, i_3], (i_3, i_2], (i_2, i_1], (i_1, j] \}$ must be the only fair partition for $(i_4,j]$ with the last three intervals being $\{(i_3, i_2], (i_2, i_1], (i_1, j] \}$, which implies the third criterion in Lemma \ref{lem:dp} is met.
\end{proof}
}

Hence, for a general interval $(i,j]$, to compute $\mathsf{LF}((i,j], i_1, i_2, i_3)$, it suffices to check for all possible values of $i_4$, each incurring two lookups of previously computed subquery results and one check of fairness in a partition for an interval of length at most $4(1+\eps)\sigma$.
\opt{arxiv}{For $\eps \in [0, 1/3)$, it can be simplified to three lookups of previously computed subquery results.}


\paragraph{Algorithm.}
\opt{arxiv}{We use Corollary~\ref{cor:small-eps-dp} and dynamic programming to compute $\bigvee_{1 \leq i_1 \leq i_2 \leq i_3 \leq n} \mathsf{LF}\Paren{[0, n], i_1, i_2, i_3}$, as follows.}
\opt{aaai}{We use Lemma~\ref{lem:dp} and dynamic programming to compute $\bigvee_{1 \leq i_1 \leq i_2 \leq i_3 \leq n} \mathsf{LF}\Paren{[0, n], i_1, i_2, i_3}$, as follows.}
Our algorithm first enumerates all possible allowable intervals $(i,j]$ as base cases (i.e., $(1-\eps)\sigma \leq j-i \leq (1+\eps)\sigma$), and computes $\mathsf{LF}((i,j], i, i, i)$ for each such $(i,j]$. Then, for general $(i,j]$, the algorithm computes $\mathsf{LF}((i,j], i_1, i_2, i_3)$ for all possible values of $i_1, i_2$, and $i_3$ given $i$ and $j$, in increasing order of $(i+j)$; this ensures all the intermediate subqueries are already computed before $\mathsf{LF}((i,j], i_1, i_2, i_3)$ is evaluated. After it computes $\mathsf{LF}((i,j], i_1, i_2, i_3)$ for all possible values of $i,j,i_1,i_2,i_3$, it examines whether $\mathsf{LF}((0,n], i_1, i_2, i_3) = \mathsf{True}$ for some $i_1, i_2, i_3$; any such true entry implies a fair partition of $[n]$, i.e., the original input. 
Note that $(0,n] = [1,n]$ is the complete set of points.

\paragraph{Running time and Enhancements.}
\opt{aaai}{Refer to the full version for a detailed analysis of the algorithm, as well as two enhancement strategies that exploit precomputation to avoid redundant computations. We present the final result:}

\opt{arxiv}{For interval $(i,j]$ to belong to the base cases, it must hold that $j \in [i+(1-\eps)\sigma, i+(1+\eps)\sigma]$, for which $i \in [0, n-(1-\eps)\sigma]$. Hence, there are $O(n\sigma)$ such pairs. For each such interval $(i,j]$, whether or not it contains a deviating group as a standalone interval can be checked naively in $O(\sigma^3)$-time. Hence all base cases can be computed within $O(n\sigma^4)$-time.
 
For every general interval $(i,j]$, the algorithm needs to enumerate all possible values of $i_1, i_2,$ and $i_3$. Note that the ranges for all possible values of $i_1$, $i_2$, and $i_3$ are exactly $2\eps \sigma$, $4\eps \sigma$, and $6\eps \sigma$; hence the number of such 3-tuples are bounded by $O(\eps^3 \sigma^3)$. For each subquery $\mathsf{LF}((i,j], i_1, i_2, i_3)$, the algorithm needs to enumerate all $O(\eps \sigma)$ possible values of $i_4$; for each $i_4$, the lookups of previously computed subqueries incur $O(1)$ computation, whereas the check for fairness for $(i_4,j]$ can be done naively in $O(\sigma^2)$-time. Hence the algorithm computes the values of all subqueries in $O(\eps^4 \sigma^4)$-time for $\eps \in [0, 1/3)$ and $O(\eps^4 \sigma^6)$-time for $\eps \geq 1/3$.
 
For the last step, the algorithm linear scans $\mathsf{LF}((0,n], i_1, i_2, i_3)$ for all $O(\eps^3 \sigma^3)$ possible values of $(i_1, i_2, i_3)$, which is $O(\eps^3 \sigma^3)$-time. Therefore, the total time complexity of the algorithm is given by 
 
\[ O(n\sigma^4) + O(\eps^4 \sigma^6) + O(\eps^3 \sigma^3) = O(n\sigma^5),\]
 
\noindent as $\sigma = O(n)$ and $\eps = O(1)$. For $\eps \in [0, 1/3)$, it is $O(n\sigma^4)$ with the second term being $O(\eps^4 \sigma^4)$.
 
\paragraph{Running time enhancement with precomputation.} 
In the algorithm above, the naive check for deviating group within a standalone interval for base cases incurs redundant computations. Instead, with standard dynamic programming techniques, one can precompute in $O(n\sigma)$-time for each allowable region $(i,j]$: (i) the number of blue points and the number of red points in $(i,j]$, i.e., $\abs{(i,j] \cap B}$ and $\abs{(i,j] \cap R}$; and thus (ii) the number of unhappy points in $(i,j]$ if $(i,j]$ is contained in a blue-majority (resp. red-majority) interval. These values can be stored using $O(n\sigma)$-memory. When checking if a deviating group $(i', j']$ exists within $(i,j]$, note that both $(i', j']$ and $(i,j]$ need to be standalone allowable intervals; this gives $i' \in (i, i+2\eps\sigma]$ and $j' \in (j-2\eps\sigma, j]$. Therefore, the check can be done in $O(\eps^2 \sigma^2)$ constant-time lookups.
 
This idea is also useful for speeding up the step of checking whether the partition $\{(i_4, i_3], (i_3, i_2], (i_2, i_1], (i_1, j] \}$ is a fair partition of $(i_4, j]$. Again using standard dynamic programming techniques, similar to the above, one can precompute in $O(n\eps^4 \sigma^4)$-time for each $(i_1, i_2, i_3, i_4, j)$ whether the above partition is fair (observe that there are $O(n)$ possibilities for $j$, then at most $2\eps\sigma$ possibilities for $i_1$, then at most $4\eps\sigma$ possibilities for $i_2$, and so on), such that the check can be replaced by $O(1)$-time lookups when evaluating each $\mathsf{LF}((i,j], i_1, i_2, i_3)$.
 
For $\eps \in [0, 1/3)$, only the first enhancement is needed, for which the total time complexity of the algorithm is improved to
\[ O(n\sigma) + O(n \eps^2 \sigma^3) + O(\eps^4 \sigma^4) + O(\eps^3 \sigma^3) = O(n \sigma^3);\]
\ 
\noindent for $\eps \in [1/3, 1/2]$, with both enhancements, the total time complexity of the algorithm is improved to
\[ O(n\sigma) + O(n\eps^4 \sigma^4) + O(n \eps^2 \sigma^3) + O(\eps^4 \sigma^4) + O(\eps^3 \sigma^3) = O(n \sigma^4),\]
 as $\sigma = O(n)$ and $\eps = O(1)$.}

\begin{theorem}
\label{thm:dp-summary}
Given an instance $(X, \sigma)$ with $|X| =n$ and $\sigma \in [n]$, and parameters $\eps \in [0, 1/2]$ and $\beta \in [1/2, 1]$, a $(\eps, \beta)$-locally fair partition of $[n]$ can be computed, or report that none exists, in time $O(n \sigma^3)$ for $\eps \in [0, 1/3)$ and $O(n \sigma^4)$ for $\eps \in [1/3, 1/2]$.
\end{theorem}


\section{Conclusion}
\label{s:conclusion}


\opt{arxiv}{We note that many of our results do extend to the fixed $k$ version. The existence results in Section~\ref{s:adversarial} extend to the fixed-$k$ case naturally.  
The dynamic program of Section~\ref{s:fair-dp} can be extended to handle a fixed $k$ in a straightforward manner.}

The main open question is extending the model to two dimensions, which poses several challenges: how to define feasible regions, how these regions tile the plane, how a deviating group/region is defined. 
Part of the difficulty stems from the fact that points are no longer linearly ordered. 
\opt{arxiv}{Although Voronoi region based methods have been proposed to construct compact regions, incorporating fairness seems very hard, and none of the existing algorithms extend to this case even to construct an approximate solution for clustered inputs. 
Currently, no polynomial-time algorithm is known even if each region of the partition is restricted to be an axis-aligned rectangle. 
In contrast, the additional freedom when defining two dimensional regions suggests a fair partition exists for a larger value of parameters assuming the input points are not concentrated along a $1$D curve.}
Despite the challenges, the lower bounds presented in the paper directly extend to two dimensions.
Additionally, the algorithm described in Section~\ref{s:fair-dp} may extend to $2$D when when the partitions considered have sufficient structure, such as when we restrict to a hierarchical partition of simple shapes, and the deviating region spans $O(1)$ regions of the partition.

\newpage
\section*{Acknowledgments}
{The authors would like to thank Hsien-Chih Chang, Brandon Fain, and anonymous reviewers for helpful discussions and feedback. This work is supported by NSF grants CCF-2113798, IIS-18-14493, and CCF-20-07556 and ONR award N00014-19-1-2268.}

\bibliographystyle{plainurl}
\bibliography{ref.bib}

\end{document}